\documentclass[12pt,english]{article}
\usepackage[T1]{fontenc}
\usepackage{amssymb}
\usepackage{amsfonts}
\usepackage{graphicx}
\usepackage{babel}

\usepackage{a4wide}

\newcommand{\nond}{N_{1},N_{2}}

\newcommand{\RR}{\mathbb{R}}

\newcommand{\btheo}{\begin{theorem}}
\newcommand{\etheo}{\end{theorem}}
\newcommand{\blem}{\begin{lemma}}
\newcommand{\elem}{\end{lemma}}
\newtheorem{theorem}{Theorem}

\newtheorem{lemma}[theorem]{Lemma}

\newenvironment{proof}[1][Proof]{\noindent\textbf{#1.} }{\ \rule{0.5em}{0.5em}}

\begin{document}

\date{\small\null\hspace*{1.64em}Date: \ \quad manuscript --- August 2004,\newline
\null\hspace*{0.6em}ArXiv --- January 2012\hspace*{0.62em}\null}

\title{Phase transitions in the time synchronization model}
\author{V. Malyshev\thanks{ Postal address: Faculty of Mathematics and Mechanics, Moscow State
University, Leninskie Gory 1, \mbox{GSP-1},  119991, Moscow, Russia. \quad
E-mail: \quad malyshev2@yahoo.com, \quad manita@mech.math.msu.su~.\newline 
\hspace*{1\parindent} 
First published in  \  Theory Probab. Appl. 50, pp. 134-141 (2006)
}\and  A. Manita}
\maketitle

\begin{abstract}
We continue the study of the time synchronization model from arXiv:1201.2141~.
There are two types \ $i=1,2$ of particles on the line $\RR$, with $N_{i}$
particles of type~$i\,$. Each particle of type $i$ moves with constant
velocity $v_{i}$. Moreover, any particle of type $i=1,2$ jumps to any
particle of type $j=1,2$ with rates $N_{j}^{-1}\alpha _{ij}$. We find phase
transitions in the clusterization (synchronization) behaviour of this system
of particles on different time scales $t=t(N)$ relative to $N=N_{1}+N_{2}$.



\end{abstract}

\section{The Model and Main result}

The simplest formulation of the model, we consider here, is in terms of the
particle system. \ On the real line there are $N_{1}$ particles of type $1$ and
$N_{2}$ particles of type $2$, $N=N_{1}+N_{2}$. Each particle of type $i=1,2$
performs two independent movements. First of all, it moves with constant speed
$v_{i}$ in the positive direction. \ We assume further that $v_{i}$ are constant
and different, thus we can assume without loss of generality that $0\leq
v_{1}<v_{2}$. The degenerate case $v_{1}=v_{2}$ is different and will be
considered separately.

Secondly, at any time interval $\left[t,t+dt\right]$ each particle of type $i $
independently of the others with probability $\alpha _{ij}dt$ decides to make a
jump to some particle of type $j$ and chooses the coordinate of the $j $-type
particle, where to jump, among the particles of type $j$, with probability
$\frac{1}{N_{j}}$. Here $\alpha _{ij}$ are given nonnegative parameters for
$i,j=1,2$. Further on, unless otherwise stated, we assume that
$\alpha _{11}=\alpha _{22}=0$, $\alpha _{12},\alpha _{21}>0$.

After such instantaneous jump the particle of type $i$ continues the
movement with the same velocity $v_{i}$. \ This defines continuous time
Markov chain
\begin{equation}\label{eq:xin1n2}
\xi _{N_{1},N_{2}}(t)=\left( x_{1}^{(1)}(t),\ldots
,x_{N_{1}}^{(1)}(t);x_{1}^{(2)}(t),\ldots ,x_{N_{2}}^{(2)}(t)\right) ,
\end{equation}
where $x_{k}^{(i)}(t)$ is the coordinate of $k$-th particle of type $i$ at time $t$.
 We assume that the initial coordinates $x_{k}^{(i)}(0)$ of the particles
at time $0$ are given. We are interested in the long time evolution of this
system on various scales with $N\rightarrow \infty $, $t=t(N)\rightarrow \infty $.

In different terms, this can be interpreted as the time synchronization problem.
In general, time synchronization problem can be presented as follows. There are
$N$ systems (processors, units, persons etc.) There is an absolute (physical)
time $t$, but each processor $j$ fulfills a~homogeneous job in its own proper
time $t_{j}=v_{j}t$, $v_{j}>0$. Proper time is measured by the amount $v_{j}$ of
the job, accomplished by the processor for the unit of the physical time, if it
is disjoint from other processors. However, there is a communication between
each pair of processors, which should lead to drastic change of their proper
times. In our case the coordinates $x_{k}^{(i)}(t)$ can be interpreted as the
modified proper times of the particles-processors, the nonmodified proper time
being $x_{k}^{(i)}(0)+v_{i}t$.

There can be many variants of exact formulation of such problem, see
\cite{GrMaPo,ManSch,MitMit}. We will call the model considered here the basic
model, because there are no restrictions on the jump process. Many other
problems include such restrictions, for example, only jumps to the left are
allowed. Due to absence of restrictions, this problem, as we will see below, is
a ``linear problem'' in the sense that after scalings it leads to linear
equations. In despite of this it has nontrivial behaviour, one sees different
picture on different time scales.

There are, however, other interesting interpretations of this model, related to
psychology, biology and physics. For example, in social psychology perception of
time and life tempo strongly depends on the social contacts and intercourse. We
will not enter the details here.

We show that the process consists of three consecutive stages: initial
desynchronization up to the critical scale, critical slow down of
desynchronization and final stabilization.


Introduce the empirical means (mass centres) and the empirical variances
\[
\overline{x^{(i)}}(t)=\frac{1}{N_{i}}\sum_{k=1}^{N_{i}}x_{k}^{(i)}(t),
\qquad
S_{i}^{2}(t)=\frac{1}{N_{i}}\sum_{k=1}^{N_{i}}\left( x_{k}^{(i)}(t)-\overline{x^{(i)}}(t)\right) ^{2}
\]
for types 1 and 2 and their means
\[
\mbox{\boldmath$\mu$}_{i}(t)=\mathsf{E}\overline{x^{(i)}}(t),\quad
\mbox{\boldmath $l$}_{12}(t)=\mbox{\boldmath$\mu$}_{1}(t)-
\mbox{\boldmath$\mu$}_{2}(t),\quad R_{i}(t)=\mathsf{E}S_{i}^{2}(t)
\]
Our first result concerns the asymptotic behavior of the empirical means.
\begin{theorem} 
For any sequences of
triplets $(N_1,N_2,t)$ such that $\min(N_1,N_2)\rightarrow \infty$ and
$t=t(N_1,N_2)\rightarrow \infty $ the following statements hold
\label{t:l-mu} \mbox{ }
\[
\mbox{\boldmath $l$}_{12}(t)\rightarrow \frac{v_{1}-v_{2}}{\alpha
_{12}+\alpha _{21}},\qquad \frac{\mbox{\boldmath$\mu$}_{i}(t)}{t}\rightarrow
\frac{\alpha _{12}v_{2}+\alpha _{21}v_{1}}{\alpha _{12}+\alpha _{21}}
\]
\end{theorem}

Assume now that $N_{i}=[c_{i}N]$, where $c_{i}>0$, $c_{1}+c_{2}=1$.
Results of the next theorem cover all domain of asymptotical behavior of $t(N)$ as
$N\rightarrow \infty$.
\begin{theorem}
\label{t-RR} There are the following three regions of asymptotic behaviour,
uniform in $t(N)$ for sufficiently large $N$:

\begin{itemize}
\item if ${\displaystyle\frac{t(N)}{N}\rightarrow 0}$, then $R_{i}(t(N))\sim
h\varkappa _{2}t(N)$,

\item if $t=t(N)=sN$ for some $s>0$, then $R_{i}(t(N))\sim
h\,(1-e^{-\varkappa _{2}s})N$,

\item if ${\displaystyle\frac{t(N)}{N}\rightarrow \infty }$, then $R_{i}(t(N))\sim hN$,
\end{itemize}

where the constant $\varkappa _{2}>0$ is defined by the formula~(\ref{eq:kap2}), see below, and
\[
h=\frac{2\alpha _{12}\alpha _{21}\left( v_{1}-v_{2}\right) ^{2}}{\varkappa
_{2}(\alpha _{12}+\alpha _{21})^{3}}\,.
\]
\end{theorem}

\section{Proof}

\paragraph{Embedded Markov Chain}

To prove theorems~\ref{t:l-mu} and~\ref{t-RR} we will use an embedded Markov
chain. Consider the continuous time Markov chain $\xi _{N_{1},N_{2}}(t)=\xi
_{N_{1},N_{2}}(t,\omega )$, defined by~(\ref{eq:xin1n2}), and let
\[
\tau _{1}(\omega )<\tau _{2}(\omega )<\cdots<\tau _{n}=\tau _{n}(\omega )<\cdots
\]
be the random moments of particle jumps. Then $\{\tau _{n+1}-\tau
_{n}\}_{n=1}^\infty$ are i.~i.~d.~random variables exponentially distributed
with mean $\gamma _{N_{1},N_{2}}=\left( N_{1}\alpha _{12}+N_{2}\alpha
_{21}\right) ^{-1}$.

We introduce a \emph{discrete time} Markov chain $\zeta _{N_{1},N_{2}}(n)$, $n=1,2,\ldots $~,
\[
\zeta _{N_{1},N_{2}}(n,\omega )=\xi _{N_{1},N_{2}}(\tau _{n}(\omega ),\omega
)
\]
with the same state space $\RR^{N_{1}+N_{2}}$. The idea is that the asymptotic
behavior of the continuous time particle system~$\xi _{N_{1},N_{2}}(t)$ can
be reduced to the asymptotic properties of the discrete time chain $\zeta
_{N_{1},N_{2}}(n)$. Indeed, by the Law of Large Numbers
\[
\tau _{n}\sim \frac{n}{N_{1}\alpha _{12}+N_{2}\alpha _{21}}\qquad
(n\rightarrow \infty )
\]
In other words, if $n$ is large, the value $n\gamma _{N_{1},N_{2}}=n\left(
N_{1}\alpha _{12}+N_{2}\alpha _{21}\right) ^{-1}$ is asymptotically equal to
the {}\textquotedblleft physical\textquotedblright\ time~$t$ associated with
the continuous time particle system~$\xi _{N_{1},N_{2}}(t)$. Similarly to
the empirical mean and empirical variance we introduce, for the embedded
chain,

\[
X_{i}(n)=\frac{1}{N_{i}}\sum_{k=1}^{N_{i}}x_{k}^{(i)}(\tau _{n}),\qquad
D_{i}(n)=\frac{1}{N_{i}}\sum_{k=1}^{N_{i}}\left( x_{k}^{(i)}(\tau
_{n})-X_{i}(n)\right) ^{2}
\]
Evidently, $X_{i}(n)=\overline{x^{(i)}}(\tau _{n})$ and $D_{i}(n)=S_{i}^{2}(
\tau _{n})$. In the sequel we shall deal with their expected values
\[
\mu _{i}(n)=\mathsf{E}X_{i}(n),\qquad d_{i}(n)=\mathsf{E}D_{i}(n),
\]
and shall need also the notation
\[
l_{12}(n)=\mu _{1}(n)-\mu _{2}(n),\quad \quad r(n)=\mathsf{E}\left(
X_{1}(n)-X_{2}(n)\right) ^{2}
\]

\paragraph{Closed equation for empirical means}

Here we prove Theorem \ref{t:l-mu}.

The following lemma can be checked by a straightforward calculation.

\begin{lemma}
The functions $\mu _{i}(n)$ satisfy to the following closed system
\begin{eqnarray*}
\mu _{1}(n+1) &=&\mu _{1}(n)+\left[ \alpha _{12}\left( \mu _{2}(n)-\mu
_{1}(n)\right) +v_{1}\right] \gamma _{N_{1},N_{2}}+\alpha _{12}\left(
v_{2}-v_{1}\right) \gamma _{N_{1},N_{2}}^{2} \\
\mu _{2}(n+1) &=&\mu _{2}(n)+\left[ \alpha _{21}\left( \mu _{1}(n)-\mu
_{2}(n)\right) +v_{2}\right] \gamma _{N_{1},N_{2}}+\alpha _{21}\left(
v_{1}-v_{2}\right) \gamma _{N_{1},N_{2}}^{2}
\end{eqnarray*}
\end{lemma}

For $l_{12}$ the equation is also linear and closed. Namely,
\begin{eqnarray*}
l_{12}(n+1) &=&l_{12}(n)+\left[ -\left( \alpha _{12}+\alpha _{21}\right)
l_{12}(n)+\left( v_{1}-v_{2}\right) \right] \gamma _{N_{1},N_{2}}+(\alpha
_{12}+\alpha _{21})\left( v_{2}-v_{1}\right) \gamma _{N_{1},N_{2}}^{2} \\
&=&l_{12}(n)\left[ 1-\gamma _{N_{1},N_{2}}\left( \alpha _{12}+\alpha
_{21}\right) \right] +\gamma _{N_{1},N_{2}}\left( v_{1}-v_{2}\right) \left[
1-\gamma _{N_{1},N_{2}}\left( \alpha _{12}+\alpha _{21}\right) \right] ,
\end{eqnarray*}
and thus we get
\begin{eqnarray}
l_{12}(n) &=&l_{12}(0)R^{n}+\frac{v_{1}-v_{2}}{\alpha _{12}+\alpha _{21}}
\left( 1-R^{n}\right) R,  \label{eq:l12} \\
&=&\frac{v_{1}-v_{2}}{\alpha _{12}+\alpha _{21}}R+\left( l_{12}(0)-\frac{
v_{1}-v_{2}}{\alpha _{12}+\alpha _{21}}R\right) R^{n}  \label{eq:C-CRn}
\end{eqnarray}
where $R=1-\gamma _{N_{1},N_{2}}\left( \alpha _{12}+\alpha _{21}\right) $.

The variables $S(n)=\alpha _{21}\mu _{1}(n)+\alpha _{12}\mu _{2}(n)$ also
satisfy the recurrent equations
\[
S(n+1)=S(n)+\gamma _{N_{1},N_{2}}\left[ \alpha _{21}v_{1}+\alpha _{12}v_{2}
\right]
\]
Thus
\begin{equation}
S(n)=S(0)+n\frac{\alpha _{21}v_{1}+\alpha _{12}v_{2}}{N_{1}\alpha
_{12}+N_{2}\alpha _{21}}=S(0)+\left( n\gamma _{N_{1},N_{2}}\right) \cdot
(\alpha _{21}v_{1}+\alpha _{12}v_{2})  \label{mean-S}
\end{equation}

From (\ref{eq:C-CRn}) and (\ref{mean-S}) the statement of Theorem~\ref{t:l-mu} follows.

\paragraph{Empirical variances}

We would like to get closed recurrent equations for $d_{i}(n)$.

\begin{lemma}
The following identity holds

\[
\mathsf{E}\left( D_{1}(n+1)\,\left\vert \,\left(
x_{i}^{(1)}(t),x_{j}^{(2)}(t),\,i=1,\ldots ,N_{1},\,j=1,\ldots ,N_{2}\right)
,\,t\leq \tau _{n}\right. \right) =\,
\]
\[
\,=D_{1}(n)+\alpha _{12}\gamma _{N_{1},N_{2}}\left[ \frac{N_{1}-1}{N_{1}}
\left( D_{2}(n)-D_{1}(n)+\left( X_{1}(n)-X_{2}(n)\right) ^{2}\right) -\frac{2
}{N_{1}}D_{1}(n)+\,\right.
\]
\[
\left. \,+2\,\frac{N_{1}-1}{N_{1}}\left( \gamma _{N_{1},N_{2}}\left(
v_{1}-v_{2}\right) \left( X_{1}(n)-X_{2}(n)\right) +\gamma
_{N_{1},N_{2}}^{2}\left( v_{1}-v_{2}\right) ^{2}\right) \right]
\]
and a similar identity for $\mathsf{E}(D_{2}(n+1)\,|\,\ldots )$ can be
obtained by a simple exchange of indices $1\leftrightarrow 2$.
\end{lemma}

Proof of this lemma is a straightforward calculation. Taking expectations in
the above formulae we see that in the equations for $d_{i}(n)$ the term $r(n) $ is also involved.

Consider the vector $w(n)=(d_{1}(n),d_{2}(n),r(n))^{T}$. We have
\begin{equation}
w(n+1)=Aw(n)+f(n)+g,  \label{eq:wAf}
\end{equation}
where $A$ is a ($3\times 3$)-matrix, not depending on $n$, $f(n)$ is a
bounded vector function of $n$, $g$ is a constant vector
\begin{equation}
A=E+\gamma _{N_{1},N_{2}}B,  \label{eq:AEdB}
\end{equation}
\[
B=B_{1}+B_{2}=\left(
\begin{array}{rcl}
-\alpha _{12} & \alpha _{12} & \alpha _{12} \\
\alpha _{21} & \,\,-\alpha _{21}\,\, & \alpha _{21} \\
0 & 0 & -2(\alpha _{12}+\alpha _{21})
\end{array}
\right) +\left(
\begin{array}{rcl}
-\alpha _{12}/N_{1} & -\alpha _{12}/N_{1} & -\alpha _{12}/N_{1} \\
-\alpha _{21}/N_{2} & \,\,-\alpha _{21}/N_{2}\,\, & -\alpha _{21}/N_{2} \\
{\displaystyle\frac{\alpha _{12}}{N_{1}}+\frac{\alpha _{21}}{N_{2}}} & \,\,\,
{\displaystyle\frac{\alpha _{12}}{N_{1}}+\frac{\alpha _{21}}{N_{2}}\,\,\,} &
{\displaystyle\frac{\alpha _{12}}{N_{1}}+\frac{\alpha _{21}}{N_{2}}}
\end{array}
\right) ,
\]

\begin{equation}
f(n)=2\,\gamma _{N_{1},N_{2}}\,(v_{1}-v_{2})\,l_{12}(n)\,\overrightarrow{q}
_{N_{1},N_{2}},\quad \overrightarrow{q}_{N_{1},N_{2}}=\left(
\begin{array}{c}
0 \\
0 \\
1
\end{array}
\right) +\gamma _{N_{1},N_{2}}\left(
\begin{array}{c}
\alpha _{12}-\alpha _{12}/N_{1} \\
\alpha _{21}-\alpha _{21}/N_{2} \\
-(\alpha _{12}+\alpha _{21})
\end{array}
\right) ,  \label{eq:fdl12}
\end{equation}
\[
g=\gamma _{N_{1},N_{2}}^{2}(v_{1}-v_{2})^{2}\left(
\begin{array}{c}
\alpha _{12}\gamma _{N_{1},N_{2}} \\
\alpha _{21}\gamma _{N_{1},N_{2}} \\
2
\end{array}
\right) .
\]

Note that $A,f(n),g$ depend on $N_{1},N_{2}$ and on other parameters of the
model. Note that for sufficiently large $N_{1},N_{2}$ (if $\alpha _{ij}$ are
fixed) all components of the vector $\overrightarrow{q}_{N_{1},N_{2}}$ are
positive. Note also that the matrix elements of $A$
are all positive for sufficiently large $N_{1},N_{2}$.
Thus Perron-Frobenius theory is applicable.

\paragraph{Spectral properties of the matrix $A$}

It is easy to check that the matrix $B_{1}$ has three distinct eigenvalues $\lambda _{1}=-(\alpha _{12}+\alpha _{21})$, $\lambda _{2}=0$, $\lambda
_{3}=-2(\alpha _{12}+\alpha _{21})$. We will study asymptotic behavior of the model for the case
when $N_{1}=c_{1}N$, $N_{2}=c_{2}N$ and $N\rightarrow \infty $. Denote $\Delta :=c_{1}\alpha _{12}+c_{2}\alpha _{21}$, that is $\gamma
_{N_{1},N_{2}}=\left( N\Delta \right) ^{-1}$.

For large $N_{i}$ the matrix $B$ is a small perturbation of the matrix $B_{1} $
\[
B=B_{1}+\frac{1}{N}B_{2,k}=B_{1}+\frac{1}{N}\left(
\begin{array}{rcl}
-\alpha _{12}/c_{1} & -\alpha _{12}/c_{1} & -\alpha _{12}/c_{1} \\
-\alpha _{21}/c_{2} & \,\,-\alpha _{21}/c_{2}\,\, & -\alpha _{21}/c_{2} \\
{\displaystyle\frac{\alpha _{12}}{c_{1}}+\frac{\alpha _{21}}{c_{2}}} & \,\,\,
{\displaystyle\frac{\alpha _{12}}{c_{1}}+\frac{\alpha _{21}}{c_{2}}\,\,\,} &
{\displaystyle\frac{\alpha _{12}}{c_{1}}+\frac{\alpha _{21}}{c_{2}}}
\end{array}
\right) .
\]
We will use perturbation theory to get eigenvalues of $B$. For $\lambda
_{1}(N)$ and $\lambda _{3}(N)$ it is sufficient to write
\begin{equation}
\lambda _{1}(N) =-(\alpha _{12}+\alpha _{21})+\underline{O}\left( {N}^{-1}\right) ,
\label{eq:lam1}
\qquad
\lambda _{3}(N) =-2(\alpha _{12}+\alpha _{21})+\underline{O}\left( {N}^{-1}\right) ,
\label{eq:lam3}
\end{equation}
however for $\lambda _{2}(N)$ we will use the result from \cite{Kato}, that
\begin{equation}
\lambda _{2}(N)=\frac{1}{N}\left( \psi ^{\prime }B_{2,k}\phi \right) +
\underline{O}\left( \frac{1}{N^{2}}\right) ,  \label{eq:lam2}
\end{equation}
where the column vector $\phi $ is the right eigenvector of $B_{1}$ with
eigenvalue $0$, the row vector $\psi ^{\prime }$ is the left eigenvector of $B_{1}$ with eigenvalue $0$, and moreover $\psi ^{\prime }\phi =1$. One can
take
\[
\phi =\left(
\begin{array}{c}
1 \\
1 \\
0
\end{array}
\right) ,\qquad \psi ^{\prime }=\left( 1+\frac{\alpha _{21}}{\alpha _{12}},1+
\frac{\alpha _{12}}{\alpha _{21}},1\right) /Z,\qquad Z=\left( \alpha
_{12}+\alpha _{21}\right) \left( \frac{1}{\alpha _{12}}+\frac{1}{\alpha _{21}
}\right) .
\]

Substituting these values to (\ref{eq:lam2}), we get
\begin{equation}
\lambda _{2}(N)=-\frac{\varkappa _{2}}{N}+\underline{O}\left( \frac{1}{N^{2}}
\right) ,  \label{eq:lam2yavn}
\end{equation}
where
\begin{equation}
\varkappa _{2}=2Z^{-1}\left( \frac{\alpha _{21}}{c_{1}}+\frac{\alpha _{12}}{
c_{2}}\right)  \label{eq:kap2}
\end{equation}

Denote $\sigma _{1}(N)$, $\sigma _{1}(N)$, $\sigma _{1}(N)$ the eigenvalues
of the matrix~$A$. From (\ref{eq:AEdB}) and (\ref{eq:lam1})--(\ref{eq:lam2})
we have the following assertion.

\begin{lemma}
\label{l-schA}

The eigenvalues of $A$ are
\begin{eqnarray}
\sigma _{1}(N) &=&1-\frac{b_{1}}{N}+\underline{O}\left( \frac{1}{N^{2}}
\right),
\qquad\qquad
\sigma _{3}(N) \;=\;1-\frac{b_{3}}{N}+\underline{O}\left( \frac{1}{N^{2}}
\right) ,\qquad\qquad{\null}
\nonumber \\
\sigma _{2}(N) &=&1-\frac{b_{2}}{N^{2}}+\underline{O}\left( \frac{1}{N^{3}}
\right)  \label{eq:sch2}
\end{eqnarray}
for some positive constants $b_{1},b_{2},b_{3}$.
\end{lemma}

We will need also the eigenvectors of $A$, which we denote
correspondingly by $e_{1}^{(N)},e_{2}^{(N)},e_{3}^{(N)}$. It is clear that
they are also the eigenvectors of the matrix $B_{1}+\frac{1}{N}B_{2,k}$.
Using the
perturbation theory~\cite{Kato} we conclude that
$e_{1}^{(N)},e_{2}^{(N)},e_{3}^{(N)}$ are small perturbations of
the eigenvectors $e_{1},e_{2},e_{3}$ of the matrix~$B_{1}$.
Thus, calculating $e_{1},e_{2},e_{3}$,  we get
\[
e_{1}^{(N)}=\left(
\begin{array}{c}
-\alpha _{12} \\
\alpha _{21} \\
0
\end{array}
\right) +\underline{O}\left( \frac{1}{N}\right) ,\quad e_{2}^{(N)}=\left(
\begin{array}{c}
1 \\
1 \\
0
\end{array}
\right) +\underline{O}\left( \frac{1}{N}\right) ,\quad e_{3}^{(N)}=\left(
\begin{array}{c}
-\alpha _{12}^{2} \\
-\alpha _{21}^{2} \\
(\alpha _{12}+\alpha _{21})^{2}
\end{array}
\right) +\underline{O}\left( \frac{1}{N}\right) .
\]

It is clear that with~(\ref{eq:lam1}) and~(\ref{eq:lam2yavn})
it is not difficult to find explicitly the constants~$b_{i}$. We will need
only $b_{2}$:
\begin{equation}
b_{2}\,=\,\frac{\varkappa _{2}}{\Delta }\,=\,\frac{2}{Z}\cdot \frac{{
\displaystyle\frac{\alpha _{12}}{c_{1}}+\frac{\alpha _{21}}{c_{2}}}}{
c_{1}\alpha _{12}+c_{2}\alpha _{21}}\,.  \label{eq:b2-kap}
\end{equation}

\paragraph{Some lemmas on the asymptotic behaviour}

The solution of the equation (\ref{eq:wAf}) can be uniquely written as
\begin{equation}
w(n)=A^{n}w(0)+\sum_{j=1}^{n}A^{j-1}f(n-j)+(1-A)^{-1}(1-A^{n})g.
\label{eq:w3slag}
\end{equation}
The following result shows that the first and last tems in~(\ref{eq:w3slag})
do not influence the asymptotics of $w(n)$.

\begin{lemma}
The following estimates hold uniformy in $n$ and $N$
\[
\left\Vert A^{n}w(0)\right\Vert \,\leq \,Const,
\qquad
\left\Vert (1-A)^{-1}(1-A^{n})g\right\Vert \,\leq \,Const.
\]
\end{lemma}

\begin{proof}
Using the basis of eigenvectors of $A$ we can write
\(
w(0)={\displaystyle\sum_{i=1}^{3}k_{w,i}e_{i}^{(N)}}.
\)
Then
\(
A^{n}w(0)=\displaystyle\sum_{i=1}^{3}k_{w,i}\left( \sigma _{i}(N)\right) ^{n}e_{i}^{(N)}.
\)
Note that ${\displaystyle\sup_{N}\left\Vert e_{i}^{(N)}\right\Vert <\infty }$.
 Moreover, $\left\vert \sigma _{i}(N)\right\vert <1$, starting from some~$N$.
 Then the first estimate follows. To get the second estimate we write
\[
g=\gamma _{N_{1},N_{2}}^{2}\sum_{i=1}^{3}k_{g,i}^{\circ }(N)e_{i}^{(N)}
\]
with the coefficients $k_{g,i}^{\circ }(N)$, bounded in $N$. Apply the
operator $(1-A)^{-1}(1-A^{n})$ to the latter expansion and note that by
Lemma~\ref{l-schA}
\[
\gamma _{N_{1},N_{2}}^{2}\frac{1-\left( \sigma _{i}(N)\right) ^{n}}{1-\sigma
_{i}(N)}\,\leq \,\frac{Const}{N}\qquad i=1,3,
\]
and
\[
\gamma _{N_{1},N_{2}}^{2}\frac{1-\left( \sigma _{2}(N)\right) ^{n}}{1-\sigma
_{2}(N)}\,\leq \,Const\,
\]
Then we get the estimate.
\end{proof}

Now we will analyze the second term in~(\ref{eq:w3slag})
\[
V_{N}(n):={\displaystyle\sum_{j=1}^{n}A^{j-1}f(n-j),}
\]
Note that the vector function $f(n)$, defined by the formula~(\ref{eq:fdl12}
), is known explicitely with the formula~(\ref{eq:l12}).

Let $\xi _{1},\xi _{2},\xi _{3}\in \mathbb{R}$ are such that
\(
\left(
0,0,1
\right)^T =\displaystyle\sum_{i=1}^{3}\xi _{i}e_{i}.
\)
Then we have immediately that
\[
\xi _{1}=\frac{\alpha _{21}-\alpha _{12}}{(\alpha _{12}+\alpha _{21})^{2}}
\,,\quad \xi _{2}=\frac{\alpha _{12}\alpha _{21}}{(\alpha _{12}+\alpha
_{21})^{2}}\,,\quad \xi _{3}=\frac{1}{(\alpha _{12}+\alpha _{21})^{2}}\,.
\]
If $\xi _{i}(N)\in \mathbb{R}$ is the coefficient in the expansion $\overrightarrow{q}_{N_{1},N_{2}}={\displaystyle\sum_{i=1}^{3}\xi
_{i}(N)e_{i}^{(N)}}$, then obviously
\begin{equation}
\xi _{i}(N)=\xi _{i}+\underline{O}\left( {N}^{-1}\right) ,\quad i=1,2,3\,.
\label{eq:xin-xi}
\end{equation}
We have then
\(
V_{N}(n)=\displaystyle\sum_{i=1}^{3}\xi _{i}(N)\left( 2\,\gamma
_{N_{1},N_{2}}\,(v_{1}-v_{2})\,{\textstyle\sum\limits_{j=1}^{n}l_{12}(n-j)}\left( \sigma
_{i}(N)\right) ^{j-1}\right) e_{i}^{(N)}\,.
\)
By~(\ref{eq:xin-xi}), and neglecting $\underline{O}\left( {N}^{-1}\right)$,
we have that for $N\rightarrow \infty $ the asymptotics of $V_{N}(n)$
coincides with the asymptotics of the sum
\begin{equation}
V_{N}^{1}(n):=\sum_{i=1}^{3}\xi _{i}\left( 2\,\gamma
_{N_{1},N_{2}}\,(v_{1}-v_{2})\,\sum_{j=1}^{n}l_{12}(n-j)\left( \sigma
_{i}(N)\right) ^{j-1}\right) e_{i}^{(N)}\,.  \label{eq:v1N}
\end{equation}

\begin{lemma}
\label{l-po-PerrFrob}

For $n=N\theta (N)$, where $\theta (N)\rightarrow +\infty $, the asymptotics
of $V_{N}^{1}(n)$ is defined by the second term, that is the first and third
are small with respect to the second.
\end{lemma}

Remind that $\sigma _{2}(N)$ is the maximal eigenvalue of the positive
matrix~$A$. Then this corresponds to Perron-Frobenius theory.

\begin{proof}
Consider the formula~(\ref{eq:l12}). If we assume, that $l_{12}(0)<0$, then
from $v_{1}<v_{2}$ it follows that each coordinate of $f(n)$ is positive for all $n$
(we use it below).
Moreover, from~(\ref{eq:l12}) one can get that there exist constants $C_{1}>C_{2}>0$,
which do not depend on $N$ and $n$, such that
\begin{equation}
0<C_{2}<(v_{1}-v_{2})l_{12}(n)<C_{1}\quad \forall n,N  \label{eq:c1c2vl}
\end{equation}
Thus, the coefficient of~$e_{2}^{(N)}$ in the sum~(\ref{eq:v1N}) is positive
and can be {\it estimated from below\/} as
\[
2\xi _{2}C_{2}\,\gamma _{N_{1},N_{2}}\,\sum_{j=1}^{n}\left( \sigma
_{2}(N)\right) ^{j-1}=2\xi _{2}C_{2}\,\gamma _{N_{1},N_{2}}\,\frac{1-\left(
\sigma _{2}(N)\right) ^{n}}{1-\sigma _{2}(N)}\,
\]

Similarly, for $i=1,3$ the absolute values of the coefficients of~$e_{i}^{(N)}$
in the sum~(\ref{eq:v1N}) can be {\it estimated from above\/} as
\[
2\xi _{i}C_{1}\,\gamma _{N_{1},N_{2}}\,\sum_{j=1}^{n}\left( \sigma
_{i}(N)\right) ^{j-1}=2\xi _{i}C_{1}\,\gamma _{N_{1},N_{2}}\,\frac{1-\left(
\sigma _{i}(N)\right) ^{n}}{1-\sigma _{i}(N)}\,
\]

Thus, to end the proof of the lemma it is sufficient to compare the
asymptotics of the following three functions
\[
\frac{1-\left( \sigma _{1}(N)\right) ^{N\theta (N)}}{1-\sigma _{1}(N)},\quad
\frac{1-\left( \sigma _{2}(N)\right) ^{N\theta (N)}}{1-\sigma _{2}(N)},\quad
\frac{1-\left( \sigma _{3}(N)\right) ^{N\theta (N)}}{1-\sigma _{3}(N)}
\]
and to show that the first and the third are small with respect to the
second. It is convenient to consider separately two cases: ~\textbf{a)} $\theta (N)\rightarrow \infty ,\,\theta (N)/N\rightarrow 0$, \textbf{b)} $\theta (N)\geq cN$, and use Lemma~\ref{l-schA}. We omit these details.
\end{proof}

\paragraph{Asymptotic behavior of expectations of empirical variances}

Now we are ready to study asymptotics of the functions $R_{1}$ and $R_{2}$
and to prove  Theorem~\ref{t-RR}. Remind that
the intervals between jumps of the embedded chain have exponential
distribution with the mean $\gamma _{N_{1},N_{2}}=(N\Delta )^{-1}$, then for
large $N$ the connection between discrete time~$n$ of the embedded chain and
absolute time~$t$ is
\(
n\,\sim \,t/\gamma _{N_{1},N_{2}}=(N\Delta )t .
\)
Thus we can take $d_{i}(\,(N\Delta )t\,)$ instead of $R_{i}(t)$. Remind also
that $d_{1}$ and $d_{2}$ are the first and second components of the vector~$w $ correspondingly.

Now we can use the lemma~\ref{l-po-PerrFrob}, which shows that, as $t(N)\rightarrow \infty $, the asymptotics of $w(\,(N\Delta )t(N)\,)$
coincides with the asymptotics of the vector
\[
\xi _{2}\left( 2\,\gamma _{N_{1},N_{2}}\,(v_{1}-v_{2})\,\sum_{j=1}^{(N\Delta
)t(N)}l_{12}(\,(N\Delta )t(N)-j\,)\left( \sigma _{2}(N)\right) ^{j-1}\right)
e_{2}^{(N)}.
\]
As $e_{2}^{(N)}=(1,1,0)^{T}+\underline{O}\left( N^{-1}\right) $, then
\[
d_{i}(\,(N\Delta )t\,)\sim 2\,\xi _{2}\,\gamma
_{N_{1},N_{2}}\,(v_{1}-v_{2})\sum_{j=1}^{(N\Delta )t(N)}l_{12}(\,(N\Delta
)t(N)-j\,)\left( \sigma _{2}(N)\right) ^{j-1}.
\]
To find the asymptotics of this expression, we use the representation~(\ref{eq:C-CRn}), which gives
\[
l_{12}(n)=C_{1,N}^{\prime }+C_{2,N}^{\prime }R^{n},\quad \quad
C_{1,N}^{\prime }\rightarrow \frac{v_{1}-v_{2}}{\alpha _{12}+\alpha _{21}}
,\quad C_{2,N}^{\prime }\rightarrow C_{2,\infty }^{\prime }\quad
(N\rightarrow \infty ).
\]
Note that the following estimate holds uniformly in $N$ and $n$
\begin{eqnarray*}
\left\vert \gamma _{N_{1},N_{2}}\,\sum_{j=1}^{n}C_{2,N}^{\prime
}R^{n-j}\left( \sigma _{2}(N)\right) ^{j-1}\right\vert &\leq &(N\Delta
)^{-1}\left\vert C_{2,N}^{\prime }\right\vert \sum_{k=1}^{\infty
}R^{k}\,=\,(N\Delta )^{-1}\left\vert C_{2,N}^{\prime }\right\vert \cdot
\frac{1}{1-R} \\
&\leq &{Const},
\end{eqnarray*}
since $1-R=\gamma_{\nond}(\alpha _{12}+\alpha _{21})$.
Consider now the asymptotics of the following expression
\[
2\,\xi _{2}\,\gamma _{N_{1},N_{2}}\,(v_{1}-v_{2})\sum_{j=1}^{(N\Delta )t(N)}
\frac{v_{1}-v_{2}}{\alpha _{12}+\alpha _{21}}\left( \sigma _{2}(N)\right)
^{j-1}=\,2\,\xi _{2}\,\frac{1}{N\Delta }\cdot \frac{\left(
v_{1}-v_{2}\right) ^{2}}{\alpha _{12}+\alpha _{21}}\cdot \frac{1-\left(
\sigma _{2}(N)\right) ^{(N\Delta )t(N)}}{1-\sigma _{2}(N)}\,.
\]
By~(\ref{eq:sch2}) and~(\ref{eq:b2-kap})
\(
\sigma _{2}(N)=1-{\displaystyle\frac{\left( \varkappa _{2}/\Delta \right) }{
N^{2}}+\underline{O}\left( {N^{-3}}\right) }
\)
and thus the problem is reduced to the study of asymptotics of
\[
\frac{2\,\xi _{2}\left( v_{1}-v_{2}\right) ^{2}}{\varkappa _{2}\left( \alpha
_{12}+\alpha _{21}\right) }\,N\,\left( 1-\left( 1-{\displaystyle\frac{\left(
\varkappa _{2}/\Delta \right) }{N^{2}}}\right) ^{(N\Delta )t(N)}\right) .
\]
Now the theorem easily follows.

\end{document}